\documentclass[11pt]{article}
\usepackage{graphicx} 

\usepackage{lmodern}

\usepackage[a4paper, margin=1in]{geometry}

\usepackage[skins,breakable]{tcolorbox}

\usepackage[ruled,algo2e]{algorithm2e}
\usepackage{algorithm}
\usepackage{algorithmic}
\SetKwInput{KwInput}{Input}
\SetKwInput{KwOutput}{Output}

\usepackage{bm}

\usepackage[hidelinks]{hyperref}
\usepackage{tikz}
\usepackage{circuitikz}
\usepackage{subcaption}

\usepackage{float}

\usepackage{multicol}

\usepackage{amsmath}
\usepackage{amsthm}
\usepackage{amssymb}

\usepackage{bibentry}

\usepackage{thmtools}
\usepackage{thm-restate}

\allowdisplaybreaks

\newcommand{\bE}{\mathbb{E}}

\newcommand{\cB}{\mathcal{B}}
\newcommand{\cG}{\mathcal{G}}
\newcommand{\cI}{\mathcal{I}}
\newcommand{\cJ}{\mathcal{J}}

\newcommand{\polylog}{\mathrm{polylog}}

\newcommand{\MIS}{\textsf{Maximum Independent Set} }
\newcommand{\tMIS}{\textsf{Maximum Independent Set} }
\newcommand{\IS}{\textsf{Independent Set} }

\newcommand{\MCQ}{\textsf{Maximum Clique} }
\newcommand{\tMCQ}{\textsf{Maximum Clique}}

\newcommand{\VCO}{V_{\text{Common}}}
\newcommand{\GCO}{G_{\text{Certain}}}
\newcommand{\GBA}{G_{\text{Base}}}
\newcommand{\cGBA}{\cG_{\text{Base}}}

\DeclareMathOperator\supp{supp}

\DeclareMathOperator*{\argmin}{arg\,min}

\newtheorem{theorem}{Theorem}[section]
\newtheorem*{theorem*}{Theorem}
\newtheorem{lemma}[theorem]{Lemma}
\theoremstyle{definition}
\newtheorem{definition}{Definition}[section]
\newtheorem*{definition*}{Definition}
\newtheorem{proposition}{Proposition}[section]
\newtheorem{corollary}{Corollary}[theorem]
\newtheorem{claim}{Claim}

\newtheorem{fact}{Fact}

\newtheorem*{notation*}{Notation}

\newtheorem{remark}{Remark}[section]

\title{Deterministic Streaming Lower Bounds for Approximate \MCQ and \MIS}
\author{Adithya Diddapur\thanks{\texttt{ard90@cam.ac.uk}, Department of Pure Mathematics and Mathematical Statistics, University of Cambridge, UK. This work was supported by the Engineering and Physical Sciences Research Council [Grant Ref: EP/Y028732/1]. Part of this work was also done whilst at the University of Bristol.}}
\date{}

\begin{document}

\maketitle

\begin{abstract}
    We study the canonical \MCQ and \MIS problems in the one-pass edge-arrival graph streaming setting.
    Here, the edges of some input graph $G = (V,E)$ are presented one at a time (possibly including deletions), before an algorithm needs to produce either a large clique or independent set at the end of the stream, with the focus being on space complexity.
    We are interested in finding $\beta$-approximate solutions, for any $\beta \geq 1$.

    \medskip

    Previous work gave an algorithm using $\tilde{O}\left(n^2/\beta^2\right)$ bits of space, together with a corresponding $\tilde{\Omega}\left(n^2/\beta^2\right)$ two-party communication lower bound [Halld\'orsson et al., ICALP'12], seeming to resolve the problem.
    However, their algorithm crucially relies on randomness, and the best known deterministic algorithm remains a folklore derandomisation using $O\left(n^2/\beta\right)$ bits of space, leaving a (deterministic) gap of size $\tilde{O}(\beta)$.

    \medskip

    We make progress towards this gap with a strong lower bound for a restricted range of $\beta$: whenever $\beta \leq \sqrt{n}/36\log n$, any deterministic algorithm for either problem must use $\Omega\left(\frac{n^2}{\beta\cdot\log n}\right)$ bits of space.
    Our proof is via a two-party one-way communication lower bound, and highlights the power of randomness when approaching either of these problems.
\end{abstract}

\clearpage

\section{Introduction}

Let $G = (V,E)$ denote the (simple, undirected, and unweighted) input graph, with $\vert V \vert = n$.
The \MCQ problem asks for a largest possible subset of vertices, $U\subseteq V$, such that every pair of contained vertices is connected by an edge in $G$, whilst the \MIS problem asks for a largest subset with no induced edges.
These are two of the most fundamental problems in theoretical computer science.

In this work, we consider these problems in the \textbf{Edge-Arrival Graph Streaming} setting.
Here, the algorithm is given the entire vertex set in advance, and the stream consists of the edge set $E$ being presented one edge at a time in an arbitrary order.
A stream which only consists of (unique) edge arrivals is known as an \textbf{Insertion-Only} stream, whilst a stream which also contains `deletions' is known as an \textbf{Insertion-Deletion} stream.
The algorithm is only required to produce an output at the end of the stream, and, in the insertion-deletion case, the output must be valid for the final set of surviving edges.
We allow our algorithms to use unbounded computation time at any point, and only focus on their space complexities.

Previous work has focused on finding large independent sets in insertion-only streams.

One such line of work studied the problem of finding \text{maximal} independent sets - any independent set which is not strictly contained in any other.
Recent work has closed the remaining gap in our understanding here, leaving us with a tight space-pass tradeoff \cite{acgmw15, akns24} (up to $n^{o(1)}$ factors).
In particular, we know that any algorithm which uses $\tilde{O}(n)$ space\footnote{We use $\tilde{O}(\cdot)$ and $\tilde{\Omega}(\cdot)$ to hide $\polylog$ factors.} (i.e. \textit{semi-streaming} space) must make $\Omega(\log\log n)$ passes through the stream.
This is far more than our setting of a single pass, where a tight $\Omega\left(n^2\right)$ space lower bound was already known \cite{ack19, cdk19}.

Another approach has been to find so-called `combinatorial' solutions - independent sets of size $\frac{n}{\Delta + 1}$, where $\Delta$ denotes the maximum degree of $G$.
Here, we know of randomised algorithms which find independent sets of size $\frac{n}{\Delta+1}$ using space $\tilde{O}(n)$ \cite{hhls15, ack19}, or which estimate $\frac{n}{\Delta+1}$ directly \cite{cdk18}.
However, the best known deterministic algorithm could only obtain solutions of size $\tilde{O}\left(\frac{n}{\Delta^2}\right)$.
Remarkably, we now also know that these are both provably optimal, with \cite{y25} proving a strict separation, and being the first to highlight the power of randomisation towards independent set problems.

In this work we consider a third perspective, namely that of finding \textit{approximate} solutions.

For both $\MCQ$ and $\tMIS$, and any parameter $\beta \geq 1$, we say that a $\mathbf{\beta}$\textbf{-approximation} is a valid clique or independent set of size at least $\omega(G)/\beta$, or $\alpha(G)/\beta$ respectively - where $\omega(G)$ and $\alpha(G)$ denote the sizes of the largest cliques and independent sets in $G$ respectively.
This is extremely difficult in the offline setting, where even computing an $n^{1-\varepsilon}$-approximation is known to be \textsf{NP}-hard, for any constant $\varepsilon > 0$ \cite{h96, z06}, in stark contrast with the previously described perspectives which both admit simple linear-time offline algorithms.

\cite{hssw12} initiated the study of these approximation problems in the streaming setting: they gave a randomised one-pass insertion-deletion streaming algorithm using space $\tilde{O}\left(n^2/\beta^2\right)$, together with a corresponding insertion-only lower bound of $\tilde{\Omega}\left(n^2/\beta^2\right)$ for the regime of $\beta \leq n^{1-\varepsilon}$.
Focussing on this regime, this yields a tight understanding up to $\polylog(n)$ terms.

However, de-randomising this algorithm has remained an open challenge, with the best known folklore derandomisation using much more space: $O\left(n^2/\beta\right)$ (still for insertion-deletion streams).
This leaves a significant gap of size $\tilde{O}\left(\beta\right)$, which is central to this work:
\[
    \text{Can this gap of size $\tilde{O}(\beta)$ be closed?}
\]

\subsection{Our Contributions}

Our contribution is to resolve this question with a strong lower bound over a restricted range of $\beta$: this folklore derandomised algorithm is (almost) optimal whenever $\beta = O(\sqrt{n}/\log n)$.
This begins to makes explicit the power of randomisation towards either of these problems.

\begin{restatable}{theorem}{CLIQUE}\label{thm: CQ}
    For any $\beta \leq \sqrt{n}/36\log n$, any one-pass deterministic streaming algorithm which outputs a $\beta$-approximate \MCQ must use $\Omega\left(\frac{n^2}{\beta\cdot\log n}\right)$ bits of space.
\end{restatable}

\begin{restatable}{theorem}{MISr}\label{thm: IS}
    For any $\beta \leq \sqrt{n}/36\log n$, any one-pass deterministic streaming algorithm which outputs a $\beta$-approximate \MIS must use $\Omega\left(\frac{n^2}{\beta\cdot\log n}\right)$ bits of space.
\end{restatable}

Our proofs are via one-way two-party communication complexity, as is a standard approach in the streaming setting.

\subsection{Related Work}

\paragraph{Deterministic Streaming Lower Bounds.}
We would be remiss to not also discuss the work of \cite{acs22}, which presented a deterministic streaming lower bound for the related graph colouring problem.
Here, the algorithm is asked to colour the vertices of a graph such that no two vertices with the same colour are adjacent to each other, and the goal is to use as few colours as possible.
They proved that any deterministic one-pass semi-streaming space algorithm for colouring must use $\exp\left(\Delta^{\Omega(1)}\right)$ many colours, in stark contrast to the $(\Delta+1)$ colouring which can be achieved with randomness \cite{ack19}.

They did this by combining two key high-level ideas: firstly that the inputs to players can (and should) be chosen adaptively based on the messages sent, followed by a novel approach towards quantifying `how much' a graph can be compressed into a message, yielding a new compression lemma.
In particular, they observe that this adaptivity when choosing inputs in fact seems to be a requirement for proving deterministic separations, since otherwise we would end up with an input distribution to which Yao's minimax principle would apply \cite{y77}, and hence only another randomised lower bound.
The work of \cite{y25} then built on this by showing that the compression lemma of \cite{acs22} can also be directly applied to the problem of finding combinatorial independent sets - their works differed in terms of the actual input graphs constructed, and the combinatorics of these graphs.

Our work can also be seen as taking inspiration from \cite{acs22}, but much more loosely: we will also be adaptive when choosing inputs for Alice and then Bob, but with both a different compression argument, and also combinatorial construction.

\paragraph{The \textsf{Planted Clique} Problem.}

This is a canonical problem in theoretical computer science: given a $\cG_{n,1/2}$ graph, the goal is to determine whether or not a clique of size $\beta$ has been inserted into it \cite{aks98}.

This is known to admit polynomial time algorithms whenever $\beta = \Omega(\sqrt{n})$, has an information-theoretic barrier when $\beta \leq 2\log n$, and is conjectured to require quasi-polynomial time for at least one value of $\beta$ in between these ranges \cite{mrs21}.
Consequently, a line of work has focussed on understanding this barrier around $\beta = \Theta(\sqrt{n})$, and why exactly the hardness arises here.
Some of this has focussed on the sum-of-squares hierarchy, where it is known that $n^{\Omega(\log n)}$ time is (almost) optimal \cite{mpw15, bhkkmp19, p21}.

Concurrently, and more related to us, another line of work has looked into statistical algorithms over `semi-random' versions of the problem \cite{fk01, csv17, mmt20, s17}.
Whilst this seems unrelated to our work in the streaming setting - our hard instances will not be \textsf{Planted Clique} instances - it turns out that their $\sqrt{n}$ barrier is related to our restricted range of $\beta \leq \sqrt{n}/36\log n$.

We believe this to be an intrinsically interesting insight connecting the two problems, and discuss it further in Section \ref{sec: techniques}.

\paragraph{Other \IS Problems.}

Due to its status as a canonical problem, several other perspectives of \MIS have also garnered research attention in the streaming setting.
This has included work on sparse graphs \cite{blsvy18}, hypergraphs \cite{hhls15}, size estimation \cite{cdk17, cdk18, ccew23}, vertex arrival streams \cite{cdk19, s25}, and geometric streams \cite{ehr16, bcw20, ddk23, adhkn26}, which we mention as relevant context for our work.

\subsection{Our Techniques}\label{sec: techniques}

We now describe the techniques behind our lower bounds from the perspective of \tMCQ.
The extension to \MIS uses the same high level ideas, and only differs in terms of the actual constructed graphs.

On a very high level, our lower bounds use a similar setup as the lower bounds by \cite{acs22} and \cite{y25}, who gave deterministic lower bounds for vertex colouring, and finding combinatorial independent sets respectively.
Namely, we will also present a one-way communication lower bound, with player inputs chosen adaptively based on the previous inputs and messages sent.
However, unlike the aforementioned works, our communication game will only involve two parties, Alice and Bob.

\paragraph{The Lower Bound of \cite{hssw12}.}
To this end, we begin our discussion with the previous randomised $\tilde{\Omega}\left(n^2/\beta^2\right)$ lower bound of \cite{hssw12}.
The heart of their lower bound is a combinatorial construction of theirs which allows them to pack $\Theta\left(n^2/\beta^2\right)$ many cliques of size $\beta$ into the graph $G$ such that no edge is contained in `too many'.
It will be convenient for us to refer to these vertex subsets of size $\beta$ as `$\beta$-subsets'.

They then used this to give a reduction from the canonical \textsf{Set-Disjointness} problem \cite{ks92}: each $\beta$-subset contained a gadget encoding a single bit from each player, with the gadget containing either a clique of size $\beta$ or of size $O(\log n)$, depending on the \textsf{Set-Disjointness} bits being encoded.

Another way to think of this is a graph construction with $\Theta(n^2/\beta^2)$ gadgets, each of which requires $\Omega(1)$ bits of communication to solve.
Building on this approach is how we will obtain our deterministic lower bound.
We still use the same packing to combine multiple gadgets, but now replace each gadget with a new one which will instead require $\Omega(\beta)$ bits of communication to solve deterministically.

\paragraph{Our Approach.}
We are now ready to outline our approach, which proceeds via two high level steps.

First, we construct a family of hard input graphs  on $n$ vertices, and partitions of edges between Alice and Bob such that, when the message conveys only a small amount of information, the graph will contain a clique of size at least $(1/2-o(1))\cdot n$, whilst Bob is only able to produce outputs of size at most $3\cdot\log n$.
We do this by using an information theoretic approach to understand Bob's perspective of Alice's input graph given the message $m$, i.e. the compression step of any protocol, and then using this to adaptively construct a hard input for Bob.
This allows us to prove that any better-than $\left(\frac{n}{12\cdot\log n}\right)$-approximation requires at least one message of size $\Omega(n)$, and is precisely the construction that we think of as a `single gadget'.

More specifically, we first sample an Erd\H{o}s-R\'{e}nyi random graph, $\GBA$, and give Alice the graph $E_A = (V,E(\GBA[A]))$ for a randomly sampled subset $A\subseteq V$ of size $\vert A \vert = n/2$.
Then, to construct Bob's input, upon receiving the message $m$ we first identify the common core of vertices which Bob is certain are contained in $A$, denoted $\VCO$.
Next, and key to our construction, Bob's input is constructed by randomly picking any $A'\in m^{-1}$ (with positive probability that $A' = A$), and then giving Bob edges induced on this subset, but \textit{avoiding} the common core.
This prevents Bob from outputting any large cliques since they will never be sure if the required corresponding edges are held by Alice, whilst still maintaining a positive probability that $G$ contains a large clique, ruling out better-than $n/(12\cdot\log n)$-approximations using space $o(n)$.
To aid in understanding, we illustrate the construction of a single gadget in Figure \ref{fig: single cq}.

The technical crux of analysing this construction is to relate the communication cost of any protocol with both $\vert\VCO\vert$ and the resulting combined input graph distribution.
We do this via an information theoretic approach to counting, and this part of our argument can be seen as a key technical contribution of this work.

We then extend this construction to obtain our final result for $\beta$-approximations by combining $\Theta\left(n^2/\beta^2\right)$ many independent gadgets.
This is precisely where we return to the edge-disjoint packing previously used by \cite{hssw12}.

To this end, for any $\beta = O(n^{1-\varepsilon})$, they showed the existence of a packing such that each pair of vertices is contained in at most $O(1/\varepsilon)$ many cliques.
Importantly for us, this falls to a single clique when $\varepsilon \geq 1/2$.
This is equivalent to saying that any two of these cliques of size $\beta$ share at most one vertex in common.
This is also precisely the condition we will need to ensure our instances remain independent, and don't leak information about each other, which is exactly why our bound only holds for the restricted range of $\beta = O(\sqrt{n})$.
The $\log n$ factor appears separately as an artifact of a re-parametrisation at the end of the proof.

We remark that this is also precisely where our work can be connected to the $\sqrt{n}$ barrier for the \textsf{Planted Clique} problem.
The (almost) tight work of \cite{s17} for that problem also relies on constructing subsets of vertices with pairwise intersections of size at most one.
We refer the reader to \cite{ms17} for a further discussion of this packing limitation.

We use this packing by scaling our gadgets down to lie within $18\cdot\beta$ many vertices, so we can still pack $\Theta(n^2/\beta^2)$ many into the graph $G$.
This involves giving Alice $\Theta\left(n^2/\beta^2\right)$ many disjoint independent single gadget inputs, before adaptively identifying the instance for which the message $m$ conveys the least information, using an information-theoretic notion of what this means. 
Then, Bob's input will be precisely the input they would be given when only considering this single gadget.

The result is an input graph distribution such that as long as there exists \textit{any} single gadget for which the message does not convey much information, the input graph may contain (with positive probability over the choice of $A'$) cliques of size $8\cdot\beta$, whilst Bob will only be able to produce outputs of size at most $6\cdot\log n$.
Thus, to obtain a better than $\left(\frac{\beta}{\log n}\right)$-approximation, the message must always convey enough information to simultaneously resolve each independent gadget which is sufficient to complete the argument since each gadget requires $\Omega(\beta)$ bits of communication to solve.

\begin{remark}
    The previous discussion of \cite{acs22} regarding the need for adaptive adversaries when proving deterministic separations seemingly precludes a proof via a standard reduction from most communication problems.
    However, to aid in intuition, we point out that our ideas have some parallels with the \textsf{Equality} communication problem, where the players need to determine whether or not they hold identical bit strings.
    This is a communication problem which is known to require $\Omega(n)$ bits of communication deterministically, but admits a randomised protocol with constant advantage using only $O(1)$ bits of communication \cite{ry20}.
\end{remark}

\begin{figure}
    \begin{center}
        \subcaptionbox{Alice's input graph $E_A$. The highlighted subset of vertices denote the sampled $A\subseteq V$, and the edges with both endpoints in this subset are individually realised as per the randomly sampled base graph $\GBA$.}
        {
            \resizebox{.45\textwidth}{!}{%
            \begin{circuitikz}
            \tikzstyle{every node}=[font=\fontsize{26.2pt}{34.0pt}\selectfont]
            \node [font=\fontsize{18.2pt}{23.7pt}\selectfont, inner xsep=0.080cm, inner ysep=0.085cm, rounded corners=0.000cm] at (9.75,6) {};
            \draw  (3.75,4.25) ellipse (8.125cm and 5.875cm);
            \begin{scope}
            \clip [rotate around={226:(3.25,3.625)}] (3.25,3.625) ellipse (4.5cm and 3.375cm);
            \foreach \x in {-4.50,-4.34,...,4.50} {
              \draw[rotate around={45:(3.25,3.625)}] ([xshift=\x cm]3.25,-0.875) -- ([xshift=\x cm]3.25,8.125);
            }
            \end{scope}
            \draw [ rotate around={226:(3.25,3.625)}] (3.25,3.625) ellipse (4.5cm and 3.375cm);
            \end{circuitikz}
            }%
        }
        \hspace{0.4cm}
        \subcaptionbox{Bob's view of $E_A$ given the message $m$. Here, we suppose that there are three graphs in the pre-image $m^{-1}$, each shown via a separate ellipse. The highlighted region is only to show which was Alice's input: Bob is unable to distinguish which one actually corresponds to Alice's input graph. The intersection $\VCO$ is shown as the shaded solid grey region in the middle.}
        {
            \resizebox{.45\textwidth}{!}{%
            \begin{circuitikz}
            \tikzstyle{every node}=[font=\fontsize{18.2pt}{23.7pt}\selectfont]
            \draw  (3.75,4.25) ellipse (8.125cm and 5.875cm);
            \begin{scope}
            \clip [rotate around={226:(3.25,3.625)}] (3.25,3.625) ellipse (4.5cm and 3.375cm);
            \foreach \x in {-4.50,-4.34,...,4.50} {
              \draw[rotate around={45:(3.25,3.625)}] ([xshift=\x cm]3.25,-0.875) -- ([xshift=\x cm]3.25,8.125);
            }
            \end{scope}
            \draw [ rotate around={226:(3.25,3.625)}] (3.25,3.625) ellipse (4.5cm and 3.375cm);
            \begin{scope}
            \clip [rotate around={-44:(6.25,5.125)}] (6.25,5.125) ellipse (4.5cm and 3.375cm);
            \end{scope}
            \draw [ rotate around={-44:(6.25,5.125)}] (6.25,5.125) ellipse (4.5cm and 3.375cm);
            \begin{scope}
            \clip [rotate around={-15:(2.25,4.875)}] (2.25,4.875) ellipse (4.5cm and 3.375cm);
            \end{scope}
            \draw [ rotate around={-15:(2.25,4.875)}] (2.25,4.875) ellipse (4.5cm and 3.375cm);
            \begin{scope}
                \clip [rotate around={226:(3.25,3.625)}] (3.25,3.625) ellipse (4.5cm and 3.375cm);
                \clip [rotate around={-44:(6.25,5.125)}] (6.25,5.125) ellipse (4.5cm and 3.375cm);
                \clip [rotate around={-15:(2.25,4.875)}] (2.25,4.875) ellipse (4.5cm and 3.375cm);
                \fill[black, opacity=0.4] (3.75,4.25) ellipse (8.125cm and 5.875cm);
            \end{scope}
            \end{circuitikz}    
            }%
        }
        \\[1cm]
        \subcaptionbox{The resulting (combined) input graph $G$ when $A' = A$. Here, $E_B$ contains the remaining edges to fill in $G[A\setminus\VCO]$, and this is shown via the solid black infill. This clearly results in a clique of size (at least) $\vert A \vert - \vert \VCO \vert$.}
        {
            \resizebox{.45\textwidth}{!}{%
            \begin{circuitikz}
            \tikzstyle{every node}=[font=\fontsize{18.2pt}{23.7pt}\selectfont]
            \node [font=\fontsize{18.2pt}{23.7pt}\selectfont, inner xsep=0.080cm, inner ysep=0.085cm, rounded corners=0.000cm] at (9.75,6) {};
            \draw  (3.75,4.25) ellipse (8.125cm and 5.875cm);

            \draw [ fill={rgb,255:red,0; green,0; blue,0}, fill opacity=1, rotate around={226:(3.25,3.625)}] (3.25,3.625) ellipse (4.5cm and 3.375cm);
            \begin{scope}
                \clip [rotate around={226:(3.25,3.625)}] (3.25,3.625) ellipse (4.5cm and 3.375cm);
                \clip [rotate around={-44:(6.25,5.125)}] (6.25,5.125) ellipse (4.5cm and 3.375cm);
                \clip [rotate around={-15:(2.25,4.875)}] (2.25,4.875) ellipse (4.5cm and 3.375cm);
                \fill[white, opacity=0.6] (3.75,4.25) ellipse (8.125cm and 5.875cm);
            \end{scope}
            \begin{scope}
            \clip [rotate around={226:(3.25,3.625)}] (3.25,3.625) ellipse (4.5cm and 3.375cm);
            \foreach \x in {-4.50,-4.34,...,4.50} {
              \draw[rotate around={45:(3.25,3.625)}] ([xshift=\x cm]3.25,-0.875) -- ([xshift=\x cm]3.25,8.125);
            }
            \end{scope}
            \begin{scope}
            \clip [rotate around={-44:(6.25,5.125)}] (6.25,5.125) ellipse (4.5cm and 3.375cm);
            \end{scope}
            \draw [ rotate around={-44:(6.25,5.125)}] (6.25,5.125) ellipse (4.5cm and 3.375cm);
            \begin{scope}
            \clip [rotate around={-15:(2.25,4.875)}] (2.25,4.875) ellipse (4.5cm and 3.375cm);
            \end{scope}
            \draw [ rotate around={-15:(2.25,4.875)}] (2.25,4.875) ellipse (4.5cm and 3.375cm);

            \end{circuitikz}
            }      
        }
        \hspace{0.4cm}
        \subcaptionbox{The resulting (combined) input graph $G$ when $A'\neq A$. Here, no edge in $E_B$ is incident to any vertex in $\VCO$, and so the result is a graph where Bob is not certain about the locations of any large cliques if they exist (with high probability over $\GBA$).}
        {
            \resizebox{.45\textwidth}{!}{%
            \begin{circuitikz}
            \tikzstyle{every node}=[font=\fontsize{18.2pt}{23.7pt}\selectfont]
            \node [font=\fontsize{18.2pt}{23.7pt}\selectfont, inner xsep=0.080cm, inner ysep=0.085cm, rounded corners=0.000cm] at (9.75,6) {};
            \draw  (3.75,4.25) ellipse (8.125cm and 5.875cm);
            \begin{scope}
            \clip [rotate around={226:(3.25,3.625)}] (3.25,3.625) ellipse (4.5cm and 3.375cm);
            \foreach \x in {-4.50,-4.34,...,4.50} {
              \draw[rotate around={45:(3.25,3.625)}] ([xshift=\x cm]3.25,-0.875) -- ([xshift=\x cm]3.25,8.125);
            }
            \end{scope}
            \draw [rotate around={226:(3.25,3.625)}] (3.25,3.625) ellipse (4.5cm and 3.375cm);
            \begin{scope}
            \clip [rotate around={-44:(6.25,5.125)}] (6.25,5.125) ellipse (4.5cm and 3.375cm);
            \foreach \x in {-4.50,-4.34,...,4.50} {
              \draw[rotate around={45:(6.25,5.125)}] ([xshift=\x cm]6.25,0.625) -- ([xshift=\x cm]6.25,9.625);
            }
            \end{scope}
            \draw [ rotate around={-44:(6.25,5.125)}] (6.25,5.125) ellipse (4.5cm and 3.375cm);
            \draw [ rotate around={-15:(2.25,4.875)}] (2.25,4.875) ellipse (4.5cm and 3.375cm);
            \end{circuitikz}
            }%
        }
    \end{center}
    \captionsetup{width=.8\linewidth}
    \caption{A step-by-step illustration of how a single gadget is produced for use in our approximate \MCQ lower bound. We use hatching to depict edges which are `partially filled in' as per the randomly sampled base graph $\GBA$.}
    \label{fig: single cq}
\end{figure}

\subsection{Outline of this Work}
In Section \ref{sec: prelims} we present all technical preliminaries required for our later proofs.
In Section \ref{sec: clique} we present the complete proof for approximate \MCQ (Theorem \ref{thm: CQ}) by first presenting a single gadget (Section \ref{sec: cq single gadget}), before showing how multiple are combined (Section \ref{sec: cq multi gadget}).
In Section \ref{sec: MIS}, we show how this is adapted to approximate \MIS (Theorem \ref{thm: IS}).
Finally, for completeness, the (folklore) derandomised algorithm is presented in Appendix \ref{app: algs}.

\section{Preliminaries}\label{sec: prelims}

We now present the required technical preliminaries for this work.

All logarithms are stated in base two.
The following standard fact regarding random graphs will be useful; we use $\cG_{n,1/2}$ to refer to the distribution of Erd\H{o}s-R\'{e}nyi random graphs (i.e. each edge is independently realised with probability $1/2$).

\begin{proposition}\cite{fk15}\label{prop: random graph clique}
    Let $G\sim\cG_{n,1/2}$.
    Then, $\Pr(\omega(G) \leq 3\cdot\log n) \geq 1 - 1/n$.
\end{proposition}

Notationally, we will use $\overline{G}$ to denote the edge-complement of $G$, and $G[U]$ to denote the subgraph of $G$ induced by $U\subseteq V$.

\subsection{Communication Complexity}\label{sec: cc prelims}

We use the standard one-way two-party communication complexity setting, first introduced by \cite{y79}.
We briefly outline this now, but also refer the reader to the excellent textbook by Rao and Yehudayoff \cite{ry20} for further discussion on the topic.

Here, we have two parties that we denote Alice and Bob, and an input graph $G = (V,E)$.
Both parties are given the entire vertex set $V$, and then the edge set is partitioned between the two.
We use $E_A$ to denote the graph given to Alice, and $E_B$ to denote that given to Bob.
Both parties only have visibility over the edges they are given, and not those held by the other party.
Then, to play the game, Alice is allowed to send some message $M = M(E_A)$ to Bob, before Bob is required to produce an output for the (shared) input graph $G$.
The goal is for the message to be as small as possible, and to this end both parties are given unlimited local computational resources.

The manner in which Alice computes their message is referred to as the \textbf{protocol}, and the \textbf{communication cost} of a protocol is the worst case (largest) message size over all inputs.
The \textbf{communication complexity} of a problem is defined as the minimum communication cost over all protocols.
It is a standard fact that communication lower bounds yield streaming lower bounds.

\begin{fact}[cf. \cite{ams96}]\label{fact: cc vs streaming}
    The communication complexity of a problem in the one-way two-party communication setting is a lower bound for the space required by any one-pass streaming algorithm for the same problem.
\end{fact}

\subsection{Counting with Information Theory}

We now present the basic information theory required by our compression arguments.
We refer the reader to the excellent textbook by Cover and Thomas \cite{cv05} for a more complete introduction to the topic.

We use $H(X)$ to denote the \textbf{Shannon entropy} of $X$ (in bits), and $H(X \mid Y = y)$ to denote the Shannon entropy of $X$ conditioned on the event $Y = y$.
We use $H(X \mid Y)$ to denote the Shannon entropy of $X$ conditioned on $Y$, and $\cI(X ; Y) = H(X) - H(X \mid Y)$ to denote the \textbf{mutual information} of $X$ and $Y$.
In particular, given a message $M$, and Alice's input $E_A$, we also refer to $\cI(E_A; M)$ as the information cost of the corresponding protocol, noting that this definition is specific to the one-way setting.

To motivate our use of these quantities, it is a standard fact that the information cost of a protocol is a lower bound for its communication cost.

\begin{fact}[cf. \cite{bbcr10}]\label{fact: IC vs CC}
    Let $\pi$ denote any one-way two-party communication protocol on any input distribution $\mu$.
    Let $M$ denote the message sent by Alice to Bob, and let $E_A$ denote Alice's input.
    Then,
    \[
        \cI_{\mu}(E_A ; M) \leq CC(\pi),
    \]
    where $CC(\pi)$ denotes the communication cost of $\pi$.
\end{fact}

The following three facts are also standard, and will be useful to us.
All three of them are presented in the previously mentioned textbook by Cover and Thomas \cite{cv05}.

\begin{fact}\label{fact: entropy facts}
    Let $X$ and $Y$ be random variables on the same support.
    Then, the following are all true:
    \begin{itemize}
        \item \textit{Entropy of any Distribution:} $H(X) \leq \log(\vert \supp(X) \vert)$, with equality if and only if $X$ is uniformly distributed.
        \item \textit{Conditional Entropy:} $H(X \mid Y) = \bE_{y\sim Y}[H(X \mid Y = y)]$.
        \item \textit{Chain Rule for Entropy:} Let $X_1,\dots,X_m$ be random variables on the same support. Then 
            \[
                H(X_1,\dots,X_m) = \sum_{i=1}^m H(X_i \mid X_{<i}),
            \]
            where $X_{<i} = (X_1,\dots,X_{i-1})$, and the same holds for conditional entropy.

            In particular, if $X_1,\dots,X_m$ are independent (follow a product distribution), then
            \[
                H(X_1,\dots,X_m) = \sum_{i=1}^m H(X_i).
            \]
        \item \textit{Conditioning Cannot Increase Entropy:} $H(X \mid Y) \leq H(X)$.
    \end{itemize}
\end{fact}

\begin{fact}\label{fact: binom entropy approx}
    For sufficiently large $n$,
    \[
        \frac{1}{n+1}\cdot 2^{n\cdot H_2(k/n)} \leq \binom{n}{k} \leq 2^{n\cdot H_2(k/n)},
    \]
    where $H_2(x) = -x\log(x) - (1-x)\log(1-x)$ denotes the binary entropy function.
\end{fact}

\begin{fact}[Data Processing Inequality]\label{fact: data processing inequality}
    Suppose that the conditional distribution of $Z$ given $Y$ is independent of $X$.
    Then $\cI(X ; Z) \leq \cI(Y ; Z)$.
\end{fact}

\paragraph{Message-Wise Mutual Information.}
We will also need the \textbf{Message-Wise Mutual Information}, which we define here.
It may be helpful to think of $Y$ as being a message conveying information about $X$.

\begin{definition}[Message-wise mutual information]
    Define 
    \[
        \cJ(X\mid Y=y, Z=z) = H(X\mid Z=z) - H(X\mid Y=y, Z=z).
    \]
    Note that the order of $Y$ and $Z$ are NOT interchangeable here.
\end{definition}

The following two properties of this message-wise mutual information will be useful to us.

\begin{proposition}\label{prop: I = exp J}
    $\cI(X; Y \mid Z = z) = \bE_{y\sim Y\mid Z=z}\left[\cJ(X\mid Y=y, Z=z)\right]$.
\end{proposition}
\begin{proof}
    \begin{align*}
        \bE_{y\sim Y\mid Z=z}\left[\cJ(X\mid Y=y, Z=z)\right] &= \bE_{y\sim Y\mid Z=z}\big[H(X\mid Z=z) - H(X\mid Y=y, Z=z)\big]\tag{Definition of $\cJ$.}\\
                                    &= H(X\mid Z=z) - \bE_{y\sim Y\mid Z=z}[H(X\mid Y=y, Z=z)]\tag{Linearity of expectation.}\\
                                    &= H(X\mid Z=z) - H(X\mid Y, Z=z)\tag{Fact \ref{fact: entropy facts}.}\\
                                    &= \cI(X ; Y \mid Z=z).
    \end{align*}
\end{proof}

\begin{proposition}\label{prop: cJ of product}
    If $X=(X_1,\dots,X_\ell)$ follows a product distribution conditioned on $Z=z$, then $\cJ(X\mid Y=y, Z=z) \geq \sum_{i=1}^\ell \cJ(X_i\mid Y=y, Z=z)$.
\end{proposition}
\begin{proof}
    We first expand
    \begin{align*}
        \cJ(X\mid Y=y, Z=z) &= \cJ(X_1,\dots,X_{\ell}\mid Y=y, Z=z)\\
                       &= H(X_1,\dots,X_{\ell} \mid Z=z) - H(X_1,\dots,X_{\ell}\mid Y=y, Z=z)\tag{Definition of $\cJ$.}\\
                       &= \left[\sum_{i=1}^{\ell} H(X_i\mid Z=z)\right] - H(X_1,\dots,X_{\ell}\mid Y=y, Z=z)\tag{since $(X_1,\dots,X_\ell)$ follows a product distribution conditioned on $Z=z$.}.
    \end{align*}
    We now consider the second term:
    \begin{align*}
        H(X_1,\dots,X_{\ell}\mid Y=y, Z=z) &= \sum_{i=1}^{\ell} H\left(X_i\mid X_{<i}, Y=y, Z=z\right)\tag{Chain rule for conditional entropy.}\\
                                      &\leq \sum_{i=1}^{\ell} H(X_i\mid Y=y, Z=z)\tag{Conditioning cannot increase entropy on average.},
    \end{align*}
    completing the argument.
\end{proof}

\section{Theorem \ref{thm: CQ}: \MCQ}\label{sec: clique}

We now prove our lower bound for approximate \tMCQ.

\CLIQUE*

\subsection{Step One: A Single Gadget}\label{sec: cq single gadget}

We begin with the description and analysis of a single gadget.

\begin{restatable}{lemma}{SINGLECLIQUE}\label{lem: single CQ}
    Any deterministic protocol which computes a better-than $\left(\frac{n}{12\cdot\log n}\right)$-approximate \MCQ must send at least one message of size $\Omega(n)$ bits.
\end{restatable}

\subsubsection{The Hard Communication Game}

\begin{tcolorbox}[standard jigsaw,opacityback=0,width=0.99\textwidth,breakable]\label{input: single clique}
    {\large{\underline{{\textbf{Input Distribution \ref{input: single clique}: Single-Gadget Clique Communication Game}}}}}
    \vspace{2em}

    \textbf{Initialisation:} Sample $\GBA\sim\cG_{n,1/2}$.

    \vspace{2em}

    \textbf{Alice's Input Distribution:} Let $V = [n]$.

    \begin{enumerate}
        \item Uniformly randomly sample $A\subseteq V$ such that $\vert A \vert = n/2$.
        \item Return $E_A = (V, E(G_{\text{Base}}[A]))$.
    \end{enumerate}
 
    Given that $G_{\text{Base}}$ is fixed at this stage, it will be convenient to refer to Alice's input via either $E_A$ or $A$ interchangeably.
    \vspace{1em}

    \textbf{The Compression (Message) Step:}
    Let $m$ denote the (fixed) message sent by Alice to Bob, and let $m^{-1}$ denote the pre-image of $m$, i.e. the set of Alice's inputs which all map to the message $m$.
    Then, define
    \[
        \VCO = \big\{ v\in V : \text{ $v\in A'$ for all $A'\in m^{-1}$}\big\},
    \]

    \textbf{Bob's Input Distribution:}

    \begin{enumerate}
        \item Randomly select $A'\in m^{-1}$.
        \item Return $E_B = (V,E(\overline{G_{\text{Base}}}[A'\setminus\VCO]))$.
    \end{enumerate}

    We will use $OUT$ to denote Bob's final output, and $G$ to denote the (combined) shared input graph.
\end{tcolorbox}

Throughout analysing this hard communication game, we remind the reader of Figure \ref{fig: single cq} to aid in understanding.

Our first step is to prove that the constructed graph $G$ is always a valid simple graph, and hence a valid input graph.

\begin{lemma}\label{lem: single cq G simple}
    $G$ is a simple graph.
\end{lemma}
\begin{proof}
    By construction both $E(E_A)\subseteq E\left(G_{\text{Base}}\right)$ and $E(E_B)\subseteq E\left(\overline{G_{\text{Base}}}\right)$.
    The result then follows from the fact that $E\left(G_{\text{Base}}\right)\cap E\left(\overline{G_{\text{Base}}}\right)=\emptyset$.
\end{proof}

We now proceed to analyse this hard game, and the constructed input graphs.

\subsubsection{Understanding \texorpdfstring{$\vert \VCO \vert$}{| Vcommon |} (the Compression Step)}\label{sec: single cq compression}

We begin by considering the set $\VCO$.
The heart of our compression step is to bound the size of this set.
We do so by proving that when the message has a low information content, then the resulting uncertainty regarding $A$ must cause $\vert \VCO \vert$ to also be small.

\begin{lemma}\label{lem: cJ >= VCO}
    $\cJ(A \mid M = m, \cGBA = \GBA) \geq \vert\VCO\vert - \log(n+1)$.
\end{lemma}
\begin{proof}
    We first decompose
    \begin{align}
        \cJ(A \mid M=m, \cGBA = \GBA) &= H(A\mid \cGBA = \GBA) - H(A\mid M = m, \cGBA = \GBA),\label{eq: cJ vs Vc}\\
            &= H(A) - H(A\mid M = m, \cGBA = \GBA),\tag{$A$ is independent of $\cGBA$.}\
    \end{align}
    and now consider these two terms separately.

    \begin{claim}\label{claim: H(A)}
        $H(A) \geq n - \log(n+1)$.
    \end{claim}
    \begin{proof}
        Since $A$ is uniform on a support of size $\binom{n}{n/2}$, we can write
        \begin{align*}
            H(A) &= \log\binom{n}{n/2}\tag{Fact \ref{fact: entropy facts}.}\\
                 &\geq \log\left(\frac{2^{n\cdot H_2(1/2)}}{n+1}\right)\tag{Fact \ref{fact: binom entropy approx}.}\\
                 &= n - \log(n+1)\tag{$H_2(1/2) = 1$.}.
        \end{align*}
    \end{proof}

    \begin{claim}\label{claim: H(A|M)}
        $H(A\mid M = m, \cGBA = \GBA) \leq n - \vert\VCO\vert$.
    \end{claim}
    \begin{proof}
        By viewing each element from $\VCO$ as guaranteed to belong in $A$, $n/2 - \vert\VCO\vert$ further vertices must be chosen from a ground set of size $n - \vert\VCO\vert$.
        Therefore, the conditional distribution of $A$ given both $M = m$ and $\cGBA = \GBA$ must have a support of size at most
        \[
            \binom{n-\vert\VCO\vert}{n/2-\vert\VCO\vert},
        \]
        allowing us to write
        \begin{align*}
            H(A\mid M = m, \cGBA = \GBA) &\leq \log\binom{n-\vert\VCO\vert}{n/2-\vert\VCO\vert}\tag{Fact \ref{fact: entropy facts}.}\\
                           &\leq \log\left(2^{(n-\vert\VCO\vert)\cdot H_2\left(\frac{n/2-\vert\VCO\vert}{n-\vert \VCO\vert}\right)}\right)\tag{Fact \ref{fact: binom entropy approx}.}\\
                           &\leq \log\left(2^{n-\vert\VCO\vert}\right)\tag{$H_2(x)\leq 1$ for all $x\in[0,1]$.}\\
                           &= n - \vert\VCO\vert.
        \end{align*}
    \end{proof}

    We now complete the proof by recombining Claims \ref{claim: H(A)} and \ref{claim: H(A|M)} into Equation \ref{eq: cJ vs Vc} to obtain
    \begin{align*}
        \cJ(A\mid M = m, \cGBA) &= H(A\mid \cGBA = \GBA) - H(A\mid M = m, \cGBA = \GBA)\\
                        &\geq n - \log(n+1) - n + \vert\VCO\vert\\
                        &= \vert\VCO\vert - \log(n+1).
    \end{align*}
\end{proof}

\begin{corollary}\label{cor: J vs VCO}
    If $\cJ(A \mid M = m, \cGBA) = o(n)$ bits, then $\vert\VCO\vert = o(n)$.
\end{corollary}

We now turn to consider the information contents of the messages themselves.

\begin{lemma}\label{lem: single clique average message cJ}
    If the protocol only sends messages of size $o(n)$, then 
    \[
        \Pr_{A,\cGBA}(\cJ(A \mid M = m(E_A), \cGBA = \GBA) = o(n)) \geq 1-o(1).
    \]
\end{lemma}
\begin{proof}
    First, by Facts \ref{fact: IC vs CC} and \ref{fact: data processing inequality},
    \[
        \cI_{A,\cGBA}(A ; M \mid \cGBA) \leq \bE_{\cGBA}\left[\cI_{A,\cGBA}(E_A ; M \mid \cGBA = \GBA)\right] = o(n).
    \]
    However, for a contradiction, if
    \begin{align*}
        \Pr_{M,\GBA}(\cJ(A \mid M = m(E_A), \cGBA = \GBA) = o(n)) < 1 - \Omega(1),
    \end{align*}
    then it must be that
    \[
        \Pr_{A,\GBA}(\cJ(A \mid M = m(E_A), \cGBA = \GBA) = \Omega(n)) \geq \Omega(1),
    \]
    and so
    \begin{align*}
        \cI_{A,\cGBA}(A ; M \mid \cGBA) &= \bE_{\cGBA}\left[\cI_{A\mid\cGBA=\GBA}(A ; M \mid \cGBA = \GBA)\right]\\
                   &=\bE_{A\mid\cGBA=\GBA,\cGBA}\left[\cJ(A \mid M = m(E_A), \cGBA = \GBA)\right]\tag{Via Proposition \ref{prop: I = exp J} since $M$ is a deterministic function of only $A$ and $\GBA$.}\\
                   &\geq \Omega(n) \cdot \Pr_{A,\GBA}(\cJ(A \mid M = m(E_A), \cGBA = \GBA) = \Omega(n))\\
                   &= \Omega(n)\cdot\Omega(1)\\
                   &= \Omega(n),
    \end{align*}
    contradicting that $\cI_{A}(A ; M\mid \cGBA) = o(n)$, and thus completing the proof.

    We note that the last inequality holds from observing $\cJ(A \mid M = m(E_A), \cGBA = \GBA) \geq 0$ since $A\mid\cGBA=\GBA$ is uniform over a support of size $\binom{n}{n/2}$, which $H(A\mid M=m,\cGBA=\GBA)$ cannot exceed.
\end{proof}

\begin{lemma}\label{lem: m^-1 > 1}
    If $\cJ(A \mid M = m, \cGBA = \GBA) = o(n)$ bits, then $\vert m^{-1} \vert > 1$ (conditioned on $\cGBA = \GBA$).
\end{lemma}
\begin{proof}
    We prove the contrapositive: suppose that $\vert m^{-1} \vert = 1$.
    Then, upon receiving $m$, Bob is able to infer $A$ directly with no uncertainty, and so $H(A\mid M = m, \cGBA = \GBA) = 0$.
    This then yields
    \[
        \cJ(A \mid M = m, \cGBA = \GBA) = H(A) - H(A\mid M = m, \cGBA = \GBA) = H(A) = (1-o(1))\cdot n.
    \]
\end{proof}

\begin{lemma}\label{lem: A in m^-1}
    $A\in m^{-1}$.
\end{lemma}
\begin{proof}
    The message is a deterministic function of $E_A$, which itself depends on $A$.
\end{proof}

\subsubsection{Understanding \texorpdfstring{$\omega(G)$}{omega(G)}}

We now turn to understand the constructed input graph $G$.

\begin{lemma}\label{lem: clique when A'=A}
    If $A' = A$, then $\omega(G)\geq n/2 - \vert\VCO\vert$.
\end{lemma}
\begin{proof}
    By construction,
    \[
        E_A[A\setminus\VCO] = G_{\text{Base}}[A\setminus\VCO],
    \]
    and
    \[
        E_B[A\setminus\VCO] = \overline{G_{\text{Base}}}[A\setminus\VCO].
    \]
    We now use this to write
    \begin{align*}
        E\left(G[A\setminus\VCO]\right) &= E\left(E_A[A\setminus\VCO]\right) \cup E\left(E_B[A\setminus\VCO]\right)\\
                             &= E\left(G_{\text{Base}}[A\setminus\VCO]\right) \cup E\left(\overline{G_{\text{Base}}}[A\setminus\VCO]\right)\\
                             &= E\left((G_{\text{Base}}\cup\overline{G_{\text{Base}}})[A\setminus\VCO]\right)\\
                             &= E\left(K_n[A\setminus\VCO]\right),
    \end{align*}
    which yields a clique in $G$ of size
    \[
        \vert A\setminus\VCO \vert = \vert A\vert - \vert\VCO \vert = n/2 - \vert\VCO\vert.
    \]
\end{proof}

\subsubsection{Understanding Bob's Output}

We now turn to consider the cliques which Bob is able to output.

The first step towards this is to argue that Bob's input reveals no additional information about $A$ given $m$. From an information theoretic perspective, this is equivalent to showing that $E_B$ and $E_A$ are conditionally independent given both $M = m$ and $\cGBA = \GBA$.

\begin{lemma}\label{lem: m+E_B = m}
    Let $(m+E_B)^{-1}$ denote the set of graphs which Alice may hold, given both $m$ and $E_B$.
    Then, $(m+E_B)^{-1} = m^{-1}$.
\end{lemma}
\begin{proof}
    We first observe that $(m+E_B)^{-1}\subseteq m^{-1}$, since $E_B$ can only further restrict $m^{-1}$.
    Now, for a contradiction, suppose that $(m+E_B)^{-1}\subset m^{-1}$.
    Then, by our above observation, there must exist some $E_{A'}\in m^{-1}\setminus(m+E_B)^{-1}$.
    I.e., conditioned on $m$, $E_{A'}$ is specifically ruled out by $E_B$.

    The only way for this to be the case is if either $G = E_{A'}\cup E_B$ is not a simple graph, or if $E_B$ is not a possible input for Bob to hold when Alice holds $E_{A'}$.

    However, by Lemma \ref{lem: single cq G simple}, it cannot be the case that $G$ is not simple, and we also see that Bob's input $E_B$ is a valid input for every graph Alice may hold in $m^{-1}$ since $E_B$ is constructed adaptively solely based on this common message.
    Therefore, it must be that such an $E_{A'}$ cannot exist.
\end{proof}

We now use this to argue that Bob can (almost always) only produce small outputs, noting that this is regardless of the message sent.

\begin{lemma}\label{lem: single clique bob output}
    Bob's output must be of size $\vert OUT \vert \leq 3\cdot\log n$, with probability at least $1-2/n$ over $\GBA$.
\end{lemma}
\begin{proof}
    Let $\GCO$ denote the graph consisting precisely of the edges which Bob is certain exist in $G$ given both $m$ and $E_B$.

    Our first step is to argue that Bob's output must be a clique in $\GCO$, taking the edge set of $OUT$ to be the corresponding induced edges.

    \begin{claim}\label{claim: single cq OUT < GCO}
        It must be that $E(OUT) \subseteq E\left(\GCO\right)$.
    \end{claim}
    \begin{proof}
        For a contradiction, suppose that $E(OUT)\not\subseteq E\left(\GCO\right)$, so there exists some edge $e\in E(OUT)\setminus E\left(\GCO\right)$.
        Then, since the protocol is deterministic, this necessitates the existence of an input graph $G$ such that Bob receives the same message $m$, input graph $E_B$, and produces the same output $OUT$, but where $e\notin E(G)$.
        However, this would cause Bob to produce an invalid output (i.e. one that is not a clique) on this input, yielding the desired contradiction.
    \end{proof}

    We now turn to bound the size of the largest cliques in $\GCO$.

    \begin{claim}
        $E\left(\GCO\right) \subseteq E(E_B) \cup E(\GBA[\VCO])$.
    \end{claim}
    \begin{proof}
        It is clear that Bob is certain every edge in $E(E_B)$ exists.

        We proceed by considering all remaining edges which Bob is certain exist.
        Each of these edges must be contained in $E_A$ (since they are not in $E_B$), and so we turn to consider which edges from $E_A$ Bob is certain exist.
        By Lemma \ref{lem: m+E_B = m}, we can approach this solely by considering the edges which Bob is sure exist given only the message $m$.

        By construction, every edge $e\in E(E_A)$ is of the form $e\in A\times A$, and so Bob can only be certain an edge exists if both of its endpoints are guaranteed to be elements of $A$.
        This is precisely the definition of $\VCO$, giving us that every such edge must be of the form $\VCO\times\VCO$.
        Then, since $E(E_A)\subseteq E(\GBA)$, and by the definition of induced subgraphs, this is equivalent to every such edge being an element of $E(\GBA[\VCO])$.
    \end{proof}

    We now use this to reason about the largest cliques in $\GCO$.

    \begin{claim}\label{claim: single cq omega(GCO)}
        $\omega(\GCO) \leq 3\cdot\log n$ with probability at least $1-2/n$ over $\GBA$.
    \end{claim}
    \begin{proof}
        We first see that $E(E_B)$ and $E(\GBA[\VCO])$ are vertex-disjoint sets of edges since every edge $e\in E(E_B)$ is only incident to the vertex set $A'\setminus\VCO$, whereas every edge in $E(\GBA[\VCO])$ is only incident to the vertex set $\VCO$.
        Therefore, $\GCO$ can be seen as the union of two vertex disjoint graphs and so
        \begin{align*}
            \omega(\GCO) &\leq \max\{\omega(E_B), \omega(\GBA[\VCO])\}\\
                         &\leq \max\{\omega(\overline{\GBA}), \omega(\GBA)\}\tag{$E(E_B)\subseteq E(\overline{\GBA})$ and $E(\GBA[\VCO])\subseteq E(\GBA)$.}\\
                         &\leq 3\cdot\log n,
        \end{align*}
        where the final inequality holds with probability at least $1-2/n$ over $\GBA$ by combining Proposition \ref{prop: random graph clique} with a simple union bound over $\GBA$ and $\overline{\GBA}$ since both graphs are Erd\H{o}s-R\'{e}nyi random graphs.
    \end{proof}

    Finally, combining Claims \ref{claim: single cq OUT < GCO} and \ref{claim: single cq omega(GCO)} completes the proof.
\end{proof}

\subsubsection{Bringing it all Together}

We are now ready to complete our proof for a single gadget.

\SINGLECLIQUE*

\begin{proof}
    We prove the contrapositive; consider a protocol, $\pi$, which only sends messages of size $o(n)$ bits (i.e. has communication cost $CC(\pi) = o(n)$), and suppose that Alice and Bob's inputs are constructed via Input Distribution \ref{input: single clique}.
    We will use the probabilistic method to show there exists an input on which the protocol $\pi$ fails to produce a suitable approximation.

    Since $CC(\pi) = o(n)$, by Lemma \ref{lem: single clique average message cJ}, 
    \[
        \Pr_{A,\GBA}(\cJ(A \mid M = m(E_A), \cGBA = \GBA) = o(n)) \geq 1 - o(1),
    \]
    and we proceed by conditioning on this event occurring, so, by Corollary \ref{cor: J vs VCO},
    \[
        \vert \VCO \vert = o(n).
    \]

    We now consider the input case for Bob where $A' = A$.
    This is valid for us to consider since $A\in m^{-1}$ by Lemma \ref{lem: A in m^-1}, and $A'$ is uniformly randomly chosen from $m^{-1}$ by construction.
    Thus, $\Pr(A' = A) = 1/\vert m^{-1} \vert > 0$, and so this input is possible via the probabilistic method.
    Then, when this event occurs, by Lemma \ref{lem: clique when A'=A} we get that (for sufficiently large $n$)
    \[
        \omega(G) \geq n/2 - \vert \VCO \vert = n/2 - o(n) \geq n/4.
    \]
    However, by Lemma \ref{lem: single clique bob output} it must also be the case that
    \[
        \vert OUT \vert \leq 3\cdot\log n
    \]
    with probability at least $1-2/n = 1 - o(1)$ over $\GBA$ (for sufficiently large $n$).

    Thus, via a union bound, over the randomness of $A$ and $\GBA$, it must be that simultaneously $\cJ(A \mid M = m, \cGBA = \GBA) = o(n)$, $\omega(G) \geq n/4$, and $\vert OUT \vert \leq 3\cdot\log n$ with probability at least
    \[
        1 - o(1) - o(1) = 1 - o(1) > 0.\tag{For sufficiently large $n$.}
    \]
    Therefore, by the probabilistic method, there must exist some $A$ and $\GBA$ such that all three of these conditions are met.
    Then, precisely when this is the graph sampled during the input distribution, it must be that the protocol can only output (at best) a
    \[
        \frac{n}{4}\cdot\frac{1}{3\cdot\log n} = \left(\frac{n}{12\cdot\log n}\right)\text{-approximation},
    \]
    completing the argument.
\end{proof}

This approach can trivially be extended to yield an initial bound on $\beta$-approximations, and this will be useful later on when combining multiple gadgets.

\begin{corollary}\label{cor: single beta clique}
    For any $\beta < n/12$, any deterministic protocol which computes a better-than $(\beta/\log n)$-approximate \MCQ must send a message of size $\Omega(\beta)$.
\end{corollary}
\begin{proof}
    The same approach still holds, by first arbitrarily selecting $12\cdot\beta$ many vertices, and then proceeding within this restricted induced subgraph.

    Thus, by the same arguments, any protocol which only sends messages of size $o(12\cdot\beta) = o(\beta)$ can only output at most a
    \[
        \frac{12\cdot\beta}{12\cdot\log(12\cdot\beta)} \geq \left(\frac{\beta}{\log n}\right)\text{-approximation.}
    \]
\end{proof}

\subsection{Step Two: Combining Multiple Gadgets}\label{sec: cq multi gadget}

We now show how to combine $\Theta(n^2/\beta^2)$ many of these gadgets into a single graph to obtain our final $\Omega\left(\frac{n^2}{\beta\cdot\log n}\right)$ lower bound.

\subsubsection{Packing (Mostly) Edge-Disjoint Cliques into \texorpdfstring{$G$}{G}}

The key tool for combining multiple gadgets is the following previously known combinatorial lemma, used by \cite{hssw12}.

\begin{lemma}\cite{hssw12,s17}\label{lem: beta covering}
    Let $\beta \leq \sqrt{n}/36$ with $n$ sufficiently large.
    Then, there exists subsets $B_1,\dots,B_{\ell}\subseteq V$ such that all of the following hold:
    \begin{itemize}
        \item For each $i$, $\vert B_i \vert = 18\cdot\beta$,
        \item For each $i\neq j$, $\vert B_i \cap B_j \vert \leq 1$,
        \item $\ell = \Theta\left(n^2/\beta^2\right)$.
    \end{itemize}
\end{lemma}

This precisely informs how we will combine multiple independent instances of our earlier gadget: each vertex subset $B_i$ will contain an independent instance of size $18\cdot\beta$.

\subsubsection{The Hard Communication Game}

\begin{tcolorbox}[standard jigsaw,opacityback=0,width=0.99\textwidth,breakable]\label{input: multi clique}
    {\large{\underline{\textbf{Input Distribution \ref{input: multi clique}: Multi-Gadget Clique Communication Game:}}}}
    \vspace{1em}

    \textbf{Initialisation:} Sample $\GBA\sim\cG_{n,1/2}$.

    \vspace{1em}

    \textbf{Alice's Input Distribution:} Let $V = [n]$.

    \begin{enumerate}
        \item Let $\cB = \{B_1,\dots,B_{\ell}\}$ denote the subsets of size $18\cdot\beta$ from Lemma \ref{lem: beta covering}.
        \item For each $i\in[\ell]$:
            \begin{enumerate}
                \item Randomly sample $A_i\subseteq B_i$ such that $\vert A_i \vert = 9\cdot\beta$.
                \item Let $E_A(i) = E\left(G_{\text{Base}}[A_i]\right)$.
            \end{enumerate}
        \item Return $E_A = (V,\bigcup_i E_A(i))$.
    \end{enumerate}

    \vspace{1em}

    \textbf{The Compression (Message) Step:}
    Again, let $m$ denote the (fixed) message sent by Alice to Bob, and let $m^{-1}$ denote the pre-image of $m$.
    For each $i\in[\ell]$, let $m^{-1}\vert_i$ denote the pre-image of $m$ restricted to $B_i$, i.e.
    \[
        m^{-1}\vert_i = \{ A'_i : \left(A'_1,\dots,A'_{\ell}\right)\in m^{-1} \text{ conditioned on $\cGBA = \GBA$} \},
    \]
    and
    \[
        \VCO\vert_i = \big\{ v\in V : \text{ $v\in A'$ for all $A'\in m^{-1}\vert_i$ \text{ conditioned on $\cGBA = \GBA$}}\big\}.
    \]

    \vspace{1em}

    \textbf{Bob's Input Distribution:}

    \begin{enumerate}
        \item Let $\min = \argmin_j \left\{ \cJ\left(A_j \mid M = m, \cGBA = \GBA\right) \right\}$.
        \item Randomly select $A_{\min}'\in m^{-1}\vert_{\min}$.
        \item Return $E_B = (V,E(\overline{G_{\text{Base}}}[A'_{\min}\setminus\VCO\vert_{\min}]))$.
   \end{enumerate}
\end{tcolorbox}

The hardness of this game arises from a similar compression argument as before - if only messages of size $o\left(n^2/\beta\right)$ are sent, then for most messages it must be the case that $\VCO\vert_{\min}$ is of size $o(\beta)$.
Then, in this case, the same arguments as those used to prove Lemma \ref{lem: single CQ} can be shown to apply, yielding the full argument.

Before proceeding with the formal arguments it will be notationally useful to denote the concatenation
\[
    A = (A_1,\dots,A_{\ell}).
\]

\subsubsection{The Multi-Gadget Compression Step}\label{sec: cq multi compression}

Again, we begin by first reasoning about the compression of the message.

These arguments follow almost identically to the previous single-gadget arguments of Section \ref{sec: single cq compression}, except that we now consider $A_{\min}$ instead of $A$ and $\VCO\vert_{\min}$ instead of $\VCO$.
Hence, the proofs are not repeated here, except for Lemma \ref{lem: multi cq CJ A_min o(beta)} which does not have an earlier counterpart.

\begin{lemma}\label{lem: multi gadget average message cJ}
    If the protocol only sends messages of size $o(n^2/\beta)$, then 
    \[
        \Pr_{A,\GBA}(\cJ(A \mid M = m, \cGBA = \GBA) = o(n^2/\beta)) \geq 1-o(1).
    \]
\end{lemma}

\begin{lemma}\label{lem: multi cq CJ A_min o(beta)}
    If $\cJ(A \mid M = m, \cGBA = \GBA) = o(n^2/\beta)$, then $\cJ(A_{\min} \mid M = m, \cGBA = \GBA) = o(\beta)$.
\end{lemma}
\begin{proof}
    \begin{align*}
        \cJ(A \mid M = m, \cGBA = \GBA) &= \cJ(A_1,\dots,A_{\ell} \mid M = m, \cGBA = \GBA)\\
                          &\geq \sum_{i=1}^{\ell} \cJ(A_i \mid M = m, \cGBA = \GBA)\tag{Proposition \ref{prop: cJ of product}.}\\
                          &\geq \sum_{i=1}^{\ell} \cJ(A_{\text{min}} \mid M = m, \cGBA = \GBA)\tag{Definition of $A_{\text{min}}$.}\\
                          &= \ell\cdot\cJ(A_{\text{min}} \mid M = m, \cGBA = \GBA)\\
                          &= \Theta(n^2/\beta^2)\cdot\cJ(A_{\text{min}} \mid M = m, \cGBA = \GBA),\tag{Lemma \ref{lem: beta covering}.}
    \end{align*}
    or equivalently,
    \[
        \cJ(A_{\text{min}} \mid M = m, \cGBA = \GBA) \leq \cJ(A \mid M = m, \cGBA = \GBA) / \Theta(n^2/\beta^2).
    \]
    Now, plugging in $\cJ(A \mid M = m, \cGBA = \GBA) = o(n^2/\beta)$ completes the argument.
\end{proof}

\begin{lemma}\label{lem: cJ(A_i) >= VCO_i}
    If $\cJ(A_{\min} \mid M = m, \cGBA = \GBA) = o(\beta)$, then $\vert \VCO\vert_{\min} \vert = o(\beta)$.
\end{lemma}

\begin{lemma}
    If $\cJ(A_{\min} \mid M = m, \cGBA = \GBA) = o(\beta)$, then $\vert m^{-1}\vert_{\min} \vert > 1$.
\end{lemma}

\begin{lemma}\label{lem: A_i in m^-1_i}
   For any message sent, $A_{\min}\in m^{-1}\vert_{\min}$.
\end{lemma}

\subsubsection{Understanding \texorpdfstring{$\omega(G)$}{omega(G)}}

We now, again, turn to understanding the constructed input graph $G$ itself.
We do not repeat the proof as it is still virtually identical to that of Lemma \ref{lem: clique when A'=A}.

\begin{lemma}\label{lem: multi-instance clique when A'=A}
    If $A'_{\min} = A_{\min}$, then $\omega(G)\geq \frac{18\cdot\beta}{2} - \vert \VCO\vert_{\min} \vert = 9\cdot\beta - \vert \VCO\vert_{\min} \vert$.
\end{lemma}

\begin{corollary}\label{cor: multi clique omega(G)}
    If $A'_{\min} = A_{\min}$ and $\vert \VCO\vert_{\min} \vert = o(\beta)$, then $\omega(G) \geq 8\cdot\beta$ (for sufficiently large $n$).
\end{corollary}

\subsubsection{Understanding Bob's Output}

We now, again, turn to consider the cliques which Bob is able to output.

The first step is exactly the same as before: arguing that Bob's input reveals no additional information about $A$ given both $m$ and $\cGBA$, and so we again do not repeat the proof.

\begin{lemma}\label{lem: multi-instance m+E_B = m}
    Let $(m+E_B)^{-1}$ denote the set of graphs which Alice may hold, given both $m$ and $E_B$.
    Then, $(m+E_B)^{-1} = m^{-1}$.
\end{lemma}

We now, again, use this to argue that Bob is (almost always) only able to produce small outputs, this time noting that the proof again has some similarities to that of Lemma \ref{lem: single clique bob output}, but is not exactly the same.

\begin{lemma}\label{lem: multi-instance clique bob output}
    Bob's output must be of size $\vert OUT \vert \leq 6\cdot\log n$ with probability at least $1-2/n$ over $\GBA$.
\end{lemma}
\begin{proof}
    Again, let $\GCO$ denote the graph consisting precisely of the edges which Bob is certain exist in $G$, given both $m$ and $E_B$.

    We remark that our first two steps are identical to those in the proof of Lemma \ref{lem: single clique bob output}.

    \begin{claim}\label{claim: multi cq OUT in GCO}
        It must be that $E(OUT) \subseteq E\left(\GCO\right)$.
    \end{claim}
    \begin{proof}
        For a contradiction, suppose that $E(OUT)\not\subseteq E\left(\GCO\right)$, so there exists some edge $e\in E(OUT)\setminus E(\GCO)$.
        Then, since the protocol is deterministic, this necessitates the existence of an input graph $G$ such that Bob receives the same message $m$, input graph $E_B$, and produces the same output $OUT$, but where $e\notin E(G)$.
        However, Bob would then produce an invalid output (i.e. one that is not a clique) on this input, yielding the desired contradiction.
    \end{proof}

    \begin{claim}\label{claim: multi cq bob gco}
        $E(\GCO) \subseteq E\left(E_B\right) \cup \left(\bigcup_{i\in[\ell]} E\left(\GBA\left[\VCO\vert_i\right]\right)\right)$.
    \end{claim}
    \begin{proof}
        Again, it is clear that Bob is certain every edge in $E\left(E_B\right)$ exists, and so we proceed by considering the remaining edges which Bob is certain exist.
        Naturally, each of these must be contained in $E_A$ (since they are not in $E_B$), and so we turn to consider the edges from $E_A$ which Bob is certain exist.
        By Lemma \ref{lem: multi-instance m+E_B = m}, we can approach this solely by considering the edges which Bob is sure exist given only the message $m$.

        By construction, every edge in $E_A$ is of the form $e\in A_j\times A_j$ for some $j\in[\ell]$, and so Bob can only be certain an edge exists if both of its endpoints are guaranteed to be elements of some $A_j$.
        This is precisely the definition of $\VCO\vert_j$, giving us that every such edge must be of the form $\VCO\vert_j \times\VCO\vert_j$ for some $j\in[\ell]$.
        Then, since $E(E_A)\subseteq E(\GBA)$, and by the definition of induced subgraphs, this is equivalent to every such edge being an element of $E(\GBA\left[\VCO\vert_j\right])$ for some $j\in[\ell]$.
    \end{proof}

    The proof now begins to differ from that of Lemma \ref{lem: single clique bob output} as we consider $\GCO$ in more detail.

    \begin{claim}\label{claim: multi cq gco induced A}
        $E\left(\GCO\left[A'_{\min}\setminus\VCO\vert_{\min}\right]\right) \subseteq E\left(\overline{\GBA}\left[A'_{\min}\setminus\VCO\vert_{\min}\right]\right)$.
    \end{claim}
    \begin{proof}
        Let $e = (u,v)\in E\left(\GCO\left[A'_{\min}\setminus\VCO\vert_{\min}\right]\right)$.
        Immediately, $u,v\notin\VCO\vert_{\min}$.
        Also, by the construction of each $B_i$, it must be that $u$ and $v$ are not both also contained in any other $B_i$.
        Thus, $e\notin E(\GBA\left[\VCO\vert_i\right])$ for any $i$, and so, by Claim \ref{claim: multi cq bob gco}, $e\in E(E_B) =  E\left(\overline{\GBA}\left[A'_{\min}\setminus\VCO\vert_{\min}\right]\right)$.
    \end{proof}

    \begin{claim}\label{claim: multi cq gco induced V - A}
        $E\left(\GCO\left[V\setminus(A'_{\min}\setminus\VCO\vert_{\min})\right]\right) \subseteq E\left(\GBA\left[V\setminus\left(A'_{\min}\setminus\VCO\vert_{\min}\right)\right]\right)$
    \end{claim}
    \begin{proof}
        From the proof of Claim \ref{claim: multi cq gco induced A}, no edge in $E(\GCO\left[V\setminus\left(A'_{\min}\setminus\VCO\vert_{\min}\right)\right])$ has an endpoint in $A'_{\min}\setminus\VCO\vert_{\min}$ whilst every edge in $E_B$ does, and so
        \[
            E\left(\GCO\left[V\setminus\left(A'_{\min}\setminus\VCO\vert_{\min}\right)\right]\right) \cap E\left(E_B\right) = \emptyset.
        \]
        Thus, $E\left(\GCO\left[V\setminus\left(A'_{\min}\setminus\VCO\vert_{\min}\right)\right]\right)\subseteq E(E_A)$, and so
        \begin{align*}
            E\left(\GCO\left[V\setminus\left(A'_{\min}\setminus\VCO\vert_{\min}\right)\right]\right) &\subseteq E\left(E_A\left[V\setminus\left(A'_{\min}\setminus\VCO\vert_{\min}\right)\right]\right)\\
                &\subseteq E\left(\GBA\left[V\setminus\left(A'_{\min}\setminus\VCO\vert_{\min}\right)\right]\right),
        \end{align*}
        with the final relation due to the construction of $E_A$.
    \end{proof}

    Now, to complete the proof, we observe that $A'_{\min}\setminus\VCO\vert_{\min}$ and $V\setminus\left(A'_{\min}\setminus\VCO\vert_{\min}\right)$ form a partition of $V$, and so
    \begin{align*}
        \omega(\GCO) &\leq \omega\left(\GCO[A'_{\min}\setminus\VCO\vert_{\min}]\right) + \omega\left(\GCO[V\setminus\left(A'_{\min}\setminus\VCO\vert_{\min}\right)]\right)\\
                     &\leq \omega\left(\overline{\GBA}\left[A'_{\min}\setminus\VCO\vert_{\min}\right]\right) + \omega\left(\GBA\left[V\setminus\left(A'_{\min}\setminus\VCO\vert_{\min}\right)\right]\right)\tag{Claims \ref{claim: multi cq gco induced A} and \ref{claim: multi cq gco induced V - A}.}\\
                     &\leq \omega\left(\overline{\GBA}\right) + \omega\left(\GBA\right)
    \end{align*}
    Then, since $\GBA$ and $\overline{\GBA}$ are both Erd\H{o}s-R\'{e}nyi random graphs, a union bound over Proposition \ref{prop: random graph clique} yields
    \[
        \Pr_{\GBA}\left(\omega(\GBA) \leq 3\cdot\log n \quad\text{and}\quad \omega\left(\overline{\GBA}\right) \leq 3\cdot\log n\right) \geq 1 - 2/n,
    \]
    and when this occurs we get
    \[
        \omega\left(\GCO\right) \leq 6\cdot\log n.
    \]
    The proof then follows via Claim \ref{claim: multi cq OUT in GCO}.
\end{proof}

\subsubsection{Bringing it all Together}

We are now ready to complete the proof of Theorem \ref{thm: CQ}, noting that the proof is again similar to that of Lemma \ref{lem: single CQ}.

\CLIQUE*
\begin{proof}
    We prove a bound for one-way two-party communication protocols, and again prove the contrapositive by considering a protocol $\pi$ which only sends messages of size $o\left(n^2/\beta\right)$ bits (i.e. has communication cost $CC(\pi) = o\left(n^2/\beta\right)$).
    However, we now suppose that Alice and Bob's inputs are constructed via Input Distribution \ref{input: multi clique}.

    Since $CC(\pi) = o(n^2/\beta)$, by Lemma \ref{lem: multi gadget average message cJ} we get that
    \[
        \Pr_{A,\GBA}(\cJ(A \vert M = m, \cGBA = \GBA) = o(n^2/\beta)) \geq 1-o(1),
    \]
    and we again proceed by conditioning on this event occurring.
    Then, by Lemma \ref{lem: multi cq CJ A_min o(beta)}, when this event occurs it must be that
    \[
        \cJ(A_{\text{min}} \mid M = m, \cGBA = \GBA) = o(\beta),
    \]
    and so, by Lemma \ref{lem: cJ(A_i) >= VCO_i},
    \[
        \vert \VCO\vert_{\min} \vert = o(\beta).
    \]

    We now consider the input case where $A'_{\min} = A_{\min}$, which is again valid for us to consider since we have that $A_{\min}\in m^{-1}\vert_{\min}$ by Lemma \ref{lem: A_i in m^-1_i}, and $A'_{\min}$ is still uniformly randomly chosen from $m^{-1}\vert_{\min}$.
    Therefore, again, $\Pr\left(A'_{\min} = A_{\min}\right) = 1/\vert m^{-1}\vert_{\min} \vert > 0$, and so this input is possible via the probabilistic method.
    Then, by Corollary \ref{cor: multi clique omega(G)}, when this event occurs, we again get that (for sufficiently large $n$)
    \[
        \omega(G) \geq 8\cdot\beta.
    \]

    However, by Lemma \ref{lem: multi-instance clique bob output}, it must also be the case that
    \[
        \vert OUT \vert \leq 6\cdot\log n
    \]
    with probability at least $1-2/n = 1- o(1)$ over $\GBA$ (for sufficiently large $n$).

    Thus, via a union bound over the randomness of $A$ and $\GBA$, it must be that simultaneously $\cJ(A \mid M = m, \cGBA = \GBA) = o(n^2/\beta)$, $\omega(G) \geq 8\cdot\beta$, and $\vert OUT \vert \leq 6\cdot\log n$ with probability at least
    \[
        1 - o(1) - o(1) > 0.\tag{For sufficiently large $n$.}
    \]
    Therefore, by the probabilistic method, there must exist some $A$ and $\GBA$ such that all three of these conditions are met.
    Then, precisely when this is the base graph sampled by the input distribution, it must be that the protocol can only output (at best) a
    \[
        8\cdot\beta\cdot\frac{1}{6\cdot\log n} = \left(\frac{8\cdot\beta}{6\cdot\log n}\right) > \left(\frac{\beta}{\log n}\right)\text{-approximation},
    \]
    giving us that any protocol which achieves a $\left(\frac{\beta}{\log n}\right)$-approximation must send at least one message of size $\Omega\left(n^2/\beta\right)$.
    Then, finally, rescaling $\beta' = \beta\cdot\log n$ completes the proof.
\end{proof}

\section{Theorem \ref{thm: IS}: The Extension to \MIS}\label{sec: MIS}

We now show how this approach can be extended to \tMIS.

\MISr*

The proof of this is almost identical to that of Theorem \ref{thm: CQ}, except that the constructed graphs given to Alice and Bob are slightly different.
In particular, the compression arguments remain unchanged.
Therefore, to avoid repetition, we will only present the new input distribution.

\subsection{The Hard Communication Game}

Again, we begin by sampling a base graph, $\GBA$, and then proceed via the well known fact that cliques in a graph $G$ correspond precisely to independent sets in the complement graph $\overline{G}$.
However, it will not be sufficient to simply give each player the complement of whatever graph they would have been given under Input Distribution \ref{input: multi clique}.
This is because we would no longer be able to guarantee that Alice and Bob's edge sets were disjoint, which is required to ensure that $G$ is simple.

Instead, our approach is to `remove' the edges that would have previously been given to Alice and Bob respectively.
This way, we are able to continue ensuring that the constructed input graph is a valid simple graph, but where the previous cliques now correspond to independent sets instead.

In accordance with this, the only differences in the hard input distribution, when compared with Input Distribution \ref{input: multi clique}, occur when explicitly defining $E_A$ and $E_B$ (i.e. in step three of Alice and Bob's respective constructions).

\begin{tcolorbox}[standard jigsaw,opacityback=0,width=0.99\textwidth,breakable]\label{input: multi is}
    {\large{\underline{\textbf{Input Distribution \ref{input: multi is}: Independent Set Communication Game:}}}}
    \vspace{1em}

    \textbf{Initialisation:} Sample $\GBA\sim\cG_{n,1/2}$.

    \vspace{1em}

    \textbf{Alice's Input Distribution:} Let $V = [n]$.

    \begin{enumerate}
        \item Let $\cB = \{B_1,\dots,B_{\ell}\}$ denote the subsets of size $18\cdot\beta$ from Lemma \ref{lem: beta covering}.
        \item For each $i\in[\ell]$:
            \begin{enumerate}
                \item Randomly sample $A_i\subseteq B_i$ such that $\vert A_i \vert = 9\cdot\beta$.
            \end{enumerate}
        \item Return $E_A = (V, E(\GBA)\setminus\bigcup_{i\in[\ell]} E(\GBA[A_i]))$.
    \end{enumerate}

    \vspace{1em}

    \textbf{The Compression (Message) Step:}
    Again, let $m$ denote the (fixed) message sent by Alice to Bob, and let $m^{-1}$ denote the pre-image of $m$.
    For each $i\in[\ell]$, let $m^{-1}\vert_i$ denote the pre-image of $m$ restricted to $B_i$, i.e.
    \[
        m^{-1}\vert_i = \{ A'_i : \left(A'_1,\dots,A'_{\ell}\right)\in m^{-1} \text{ conditioned on $\cGBA = \GBA$}\}.
    \]
    Finally, for each $i\in[\ell]$, define
    \[
        \VCO\vert_i = \big\{ v\in V : \text{ $v\in A'$ for all $A'\in m^{-1}\vert_i$} \text{ conditioned on $\cGBA = \GBA$}\big\}.
    \]

    \vspace{1em}

    \textbf{Bob's Input Distribution:}

    \begin{enumerate}
        \item Let $\min = \argmin_j \left\{ \cJ\left(A_j \mid M = m, \cGBA = \GBA\right) \right\}$.
        \item Randomly select $A_{\min}'\in m^{-1}\vert_{\min}$.
        \item Return $E_B = (V, E(\overline{\GBA})\setminus E(\overline{\GBA}[A'_{\text{min}}\setminus\VCO\vert_{\text{min}}]))$.
   \end{enumerate}
\end{tcolorbox}

\subsection{Bringing it all Together}

As previously mentioned, the analysis of this game is virtually identical to the multi-gadget clique game in Section \ref{sec: cq multi gadget}, and so we don't repeat it here.

\section*{Acknowledgments}

The author would like to thank Christian Konrad for introducing them to the problem, and Kheeran K. Naidu for interesting discussions.
The author would also like to thank previous anonymous reviewers for their helpful comments towards improving this manuscript.

\section*{AI Declaration}

AI tools were used to assist in the spelling and grammar of this work.
AI was used to help fix a minor counting error in an earlier version of Claims \ref{claim: multi cq gco induced A} and \ref{claim: multi cq gco induced V - A}.
It also pointed out earlier errors in the usage of $\cJ(A\mid M = m, \cGBA = \GBA)$, a previous separate proof regarding the independent set hard instance, and in the application of Lemma \ref{lem: beta covering}, which was fixed without AI assistance - the previous version considered an unrestricted range of $\beta$, which caused the proof of Lemma \ref{lem: multi-instance clique bob output} to break as pairs of vertices were contained in multiple $B_i$'s.
All other mathematical work was done without the use of AI tools.

\bibliographystyle{alpha}
\bibliography{bibliography}

\appendix

\section{The Algorithm}\label{app: algs}

We now outline the $O\left(n^2/\beta\right)$ space deterministic algorithm which can be obtained by de-randomising that of \cite{hssw12}.
We do not claim any novelty here, and only present it for completeness.

\begin{theorem}
    For any $\beta \leq n/\log n$, there exists a deterministic one-pass insertion-deletion streaming algorithm which computes either a $\beta$-approximate \MCQ or \tMIS using $O\left(n^2/\beta\right)$ bits of space.
\end{theorem}
\begin{proof}
    The algorithm can be stated simply.
    First, arbitrarily partition $V$ into at most $\beta$ many disjoint subsets, $V_1,\dots,V_{\beta}$, each of size at most $\lceil n/\beta \rceil$.
    This can always be done.
    Then, during the stream, simply store all edges (and deletions) which are induced in any one of these subsets by maintaining the adjacency matrix of each subset.
    Finally, at the end of the stream, compute an optimal solution within each $V_i$, and return the largest of these.

    By a simple averaging argument, since there are exactly $\beta$ many subsets there must exist some $V_i$ which contains at least a $1/\beta$ fraction of the vertices from some optimal solution $OPT$.
    Thus, this $V_i$ must contain a sufficiently large induced solution, which can be returned at the end if nothing larger is found.

    To bound the space usage, each $V_i$ contains $O(n/\beta)$ many vertices, and so there are at most $O\left(n^2/\beta^2\right)$ many induced pairs of vertices.
    Thus, each adjacency matrix can be stored using $O\left(n^2/\beta^2\right)$ bits of space, and its labels with $O(n\log(n)/\beta)$ bits of space.
    Then, since there are at most $\beta+1$ many subsets,
    \[
        (\beta+1) \cdot O\left(n^2/\beta^2 + n\cdot\log(n)/\beta\right) = O\left(n^2/\beta + n\cdot\log(n)\right)
    \]
    bits of space suffice to maintain each of these adjacency matrices during the stream.
    This is precisely $O(n^2/\beta)$ when $\beta < n/\log n$, and $O(n\cdot\log(n))$ otherwise.
\end{proof}

\end{document}